\documentclass[notitlepage,reqno,12pt]{amsart}
\usepackage{amsmath,amssymb,amsbsy,amsfonts,amsthm,latexsym,
            amsopn,amstext,amsxtra,euscript,amscd, fancyhdr}
\usepackage{hyperref}   
\usepackage[margin=3cm]{geometry}

\begin{document}
\bibliographystyle{plain}
\newfont{\teneufm}{eufm10}
\newfont{\seveneufm}{eufm7}
\newfont{\fiveeufm}{eufm5}
%%%%%%%%%%%%%%%%%%%%%%%%%%%%%%%%%
%
%  allow automatic size selection in math mode
%
%%%%%%%%%%%%%%%%%%%%%%%%%%%%%%%%%
\newfam\eufmfam
              \textfont\eufmfam=\teneufm \scriptfont\eufmfam=\seveneufm
              \scriptscriptfont\eufmfam=\fiveeufm
%%%%%%%%%%%%%%%%%%%%%%%%%%%%%%%%%
%%%%%%%%%%%%%%%%%%% 
\def\bbbr{{\rm I\!R}}
\def\bbbm{{\rm I\!M}}
\def\bbbn{{\rm I\!N}}
\def\bbbf{{\rm I\!F}}
\def\bbbh{{\rm I\!H}}
\def\bbbk{{\rm I\!K}}
\def\bbbp{{\rm I\!P}}
\def\bbbone{{\mathchoice {\rm 1\mskip-4mu l} {\rm 1\mskip-4mu l}
{\rm 1\mskip-4.5mu l} {\rm 1\mskip-5mu l}}}
\def\bbbc{{\mathchoice {\setbox0=\hbox{$\displaystyle\rm C$}\hbox{\hbox
to0pt{\kern0.4\wd0\vrule height0.9\ht0\hss}\box0}}
{\setbox0=\hbox{$\textstyle\rm C$}\hbox{\hbox
to0pt{\kern0.4\wd0\vrule height0.9\ht0\hss}\box0}}
{\setbox0=\hbox{$\scriptstyle\rm C$}\hbox{\hbox
to0pt{\kern0.4\wd0\vrule height0.9\ht0\hss}\box0}}
{\setbox0=\hbox{$\scriptscriptstyle\rm C$}\hbox{\hbox
to0pt{\kern0.4\wd0\vrule height0.9\ht0\hss}\box0}}}}
\def\bbbq{{\mathchoice {\setbox0=\hbox{$\displaystyle\rm
Q$}\hbox{\raise
0.15\ht0\hbox to0pt{\kern0.4\wd0\vrule height0.8\ht0\hss}\box0}}
{\setbox0=\hbox{$\textstyle\rm Q$}\hbox{\raise
0.15\ht0\hbox to0pt{\kern0.4\wd0\vrule height0.8\ht0\hss}\box0}}
{\setbox0=\hbox{$\scriptstyle\rm Q$}\hbox{\raise
0.15\ht0\hbox to0pt{\kern0.4\wd0\vrule height0.7\ht0\hss}\box0}}
{\setbox0=\hbox{$\scriptscriptstyle\rm Q$}\hbox{\raise
0.15\ht0\hbox to0pt{\kern0.4\wd0\vrule height0.7\ht0\hss}\box0}}}}
\def\bbbt{{\mathchoice {\setbox0=\hbox{$\displaystyle\rm.
T$}\hbox{\hbox to0pt{\kern0.3\wd0\vrule height0.9\ht0\hss}\box0}}
{\setbox0=\hbox{$\textstyle\rm T$}\hbox{\hbox
to0pt{\kern0.3\wd0\vrule height0.9\ht0\hss}\box0}}
{\setbox0=\hbox{$\scriptstyle\rm T$}\hbox{\hbox
to0pt{\kern0.3\wd0\vrule height0.9\ht0\hss}\box0}}
{\setbox0=\hbox{$\scriptscriptstyle\rm T$}\hbox{\hbox
to0pt{\kern0.3\wd0\vrule height0.9\ht0\hss}\box0}}}}
\def\bbbs{{\mathchoice
{\setbox0=\hbox{$\displaystyle     \rm S$}\hbox{\raise0.5\ht0\hbox
to0pt{\kern0.35\wd0\vrule height0.45\ht0\hss}\hbox
to0pt{\kern0.55\wd0\vrule height0.5\ht0\hss}\box0}}
{\setbox0=\hbox{$\textstyle        \rm S$}\hbox{\raise0.5\ht0\hbox
to0pt{\kern0.35\wd0\vrule height0.45\ht0\hss}\hbox
to0pt{\kern0.55\wd0\vrule height0.5\ht0\hss}\box0}}
{\setbox0=\hbox{$\scriptstyle      \rm S$}\hbox{\raise0.5\ht0\hbox
to0pt{\kern0.35\wd0\vrule height0.45\ht0\hss}\raise0.05\ht0\hbox
to0pt{\kern0.5\wd0\vrule height0.45\ht0\hss}\box0}}
{\setbox0=\hbox{$\scriptscriptstyle\rm S$}\hbox{\raise0.5\ht0\hbox
to0pt{\kern0.4\wd0\vrule height0.45\ht0\hss}\raise0.05\ht0\hbox
to0pt{\kern0.55\wd0\vrule height0.45\ht0\hss}\box0}}}}
\def\bbbz{{\mathchoice {\hbox{$\sf\textstyle Z\kern-0.4em Z$}}
{\hbox{$\sf\textstyle Z\kern-0.4em Z$}}
{\hbox{$\sf\scriptstyle Z\kern-0.3em Z$}}
{\hbox{$\sf\scriptscriptstyle Z\kern-0.2em Z$}}}}
\def\ts{\thinspace}

\newtheorem{theorem}{Theorem}
\newtheorem{lemma}[theorem]{Lemma}
\newtheorem{claim}[theorem]{Claim}
\newtheorem{cor}[theorem]{Corollary}
\newtheorem{prop}[theorem]{Proposition}
\newtheorem{definition}[theorem]{Definition}
\newtheorem{remark}[theorem]{Remark}
\newtheorem{question}[theorem]{Open Question}
\newtheorem{example}[theorem]{Example}
\newtheorem{problem}[theorem]{Problem}

\def\qed{\ifmmode
\squareforqed\else{\unskip\nobreak\hfil
\penalty50\hskip1em\null\nobreak\hfil\squareforqed
\parfillskip=0pt\finalhyphendemerits=0\endgraf}\fi}

\def\squareforqed{\hbox{\rlap{$\sqcap$}$\sqcup$}}

\def \C {{\mathbb C}}
\def \F {{\mathbb F}}
\def \L {{\mathbb L}}
\def \K {{\mathbb K}}
\def \Q {{\mathbb Q}}
\def \Z {{\mathbb Z}}
\def\cA{{\mathcal A}}
\def\cB{{\mathcal B}}
\def\cC{{\mathcal C}}
\def\cD{{\mathcal D}}
\def\cE{{\mathcal E}}
\def\cF{{\mathcal F}}
\def\cG{{\mathcal G}}
\def\cH{{\mathcal H}}
\def\cI{{\mathcal I}}
\def\cJ{{\mathcal J}}
\def\cK{{\mathcal K}}
\def\cL{{\mathcal L}}
\def\cM{{\mathcal M}}
\def\cN{{\mathcal N}}
\def\cO{{\mathcal O}}
\def\cP{{\mathcal P}}
\def\cQ{{\mathcal Q}}
\def\cR{{\mathcal R}}
\def\cS{{\mathcal S}}
\def\cT{{\mathcal T}}
\def\cU{{\mathcal U}}
\def\cV{{\mathcal V}}
\def\cW{{\mathcal W}}
\def\cX{{\mathcal X}}
\def\cY{{\mathcal Y}}
\def\cZ{{\mathcal Z}}
\newcommand{\rmod}[1]{\: \mbox{mod}\: #1}

\def\tcN{\cN^\mathbf{c}}
\def\F{\mathbb F}
\def\Tr{\operatorname{Tr}}
\def\mand{\qquad \mbox{and} \qquad}
\renewcommand{\vec}[1]{\mathbf{#1}}
\def\eqref#1{(\ref{#1})}
\newcommand{\ignore}[1]{}
\hyphenation{re-pub-lished}
\parskip 1.5 mm
\def\lln{{\mathrm Lnln}}
\def\Res{\mathrm{Res}\,}
\def\F{{\bbbf}}
\def\Fp{\F_p}
\def\fp{\Fp^*}
\def\Fq{\F_q}
\def\ff{\F_2}
\def\ffn{\F_{2^n}}
\def\K{{\bbbk}}
\def \Z{{\bbbz}}
\def \N{{\bbbn}}
\def\Q{{\bbbq}}
\def \R{{\bbbr}}
\def \P{{\bbbp}}
\def\Zm{\Z_m}
\def \Um{{\mathcal U}_m}
\def \Bf{\frak B}
\def\Km{\cK_\mu}
\def\va {{\mathbf a}}
\def \vb {{\mathbf b}}
\def \vc {{\mathbf c}}
\def\vx{{\mathbf x}}
\def \vr {{\mathbf r}}
\def \vv {{\mathbf v}}
\def\vu{{\mathbf u}}
\def \vw{{\mathbf w}}
\def \vz {{\mathbfz}}
\def\\{\cr}
\def\({\left(}
\def\){\right)}
\def\fl#1{\left\lfloor#1\right\rfloor}
\def\rf#1{\left\lceil#1\right\rceil}
\def\flq#1{{\left\lfloor#1\right\rfloor}_q}
\def\flp#1{{\left\lfloor#1\right\rfloor}_p}
\def\flm#1{{\left\lfloor#1\right\rfloor}_m}
\def\Al{{\sl Alice}}
\def\Bob{{\sl Bob}}
\def\Or{{\mathcal O}}
\def\inv#1{\mbox{\rm{inv}}\,#1}
\def\invM#1{\mbox{\rm{inv}}_M\,#1}
\def\invp#1{\mbox{\rm{inv}}_p\,#1}
\def\Ln#1{\mbox{\rm{Ln}}\,#1}
\def \nd {\,|\hspace{-1.2mm}/\,}
\def\ord{\mu}
\def\E{\mathbf{E}}
\def\Cl{{\mathrm {Cl}}}
\def\epp{\mbox{\bf{e}}_{p-1}}
\def\ep{\mbox{\bf{e}}_p}
\def\eq{\mbox{\bf{e}}_q}
\def\bm{\bf{m}}
\newcommand{\floor}[1]{\lfloor {#1} \rfloor}
\newcommand{\comm}[1]{\marginpar{
\vskip-\baselineskip
\raggedright\footnotesize
\itshape\hrule\smallskip#1\par\smallskip\hrule}}
\def\rem{{\mathrm{\,rem\,}}}
\def\dist {{\mathrm{\,dist\,}}}
\def\etal{{\it et al.}}
\def\ie{{\it i.e. }}
\def\veps{{\varepsilon}}
\def\eps{{\eta}}
\def\ind#1{{\mathrm {ind}}\,#1}
               \def \MSB{{\mathrm{MSB}}}
\newcommand{\abs}[1]{\left| #1 \right|}

%%%%%%%%%%%%%%%  Topmatter %%%%%%%%%%%%%%%%%%
\title[Extracting over Jacobian of Hyperelliptic curves]{Extracting a uniform random bit-string over Jacobian of Hyperelliptic curves of Genus $2$}
%
%\subjclass[2010]{Primary 11G20,11Y16 ; Secondary 11T23}
%\keywords{Elliptic curves, generators, finite fields}

\author{
{\sc Bernadette~Faye}\quad %and \quad
%{\sc Florian~Luca}
}

\address{
Universit\'e Cheikh Anta Diop de Dakar \newline
Departement de Mathematiques et d'Informatique\newline
BP: 5005, Dakar-Fann\newline
Dakar, Senegal.
}

\address{
African Institute for Mathematical Sciences(AIMS) \newline
Km 2 route de Joal \newline
BP 1418 , Mbour, Senegal and \newline 
School of Mathematics\newline 
University of the Witwatersrand \newline
Private Bag X3, Wits 2050, South Africa.
}

\email{bernadette@aims-senegal.org}
%\date{\today}
\pagenumbering{arabic}

\begin{abstract} Here, we proposed an improved version of the deterministic random extractors $SEJ$ and $PEJ$ proposed by R. R. Farashahi in \cite{F} in 2009.  By using the Mumford's representation of a reduced divisor $D$ of the Jacobian $J(\mathbb{F}_q)$ of a hyperelliptic curve $\mathcal{H}$ of genus $2$ with odd characteristic, we extract  a perfectly random bit string of the sum of abscissas of rational points on $\mathcal{H}$ in the support of $D$. By this new approach, we reduce in an elementary way the upper bound of the statistical distance of the deterministic randomness extractors defined over $\mathbb{F}_q$ where $q=p^n$, for some positive integer $n\geq 1$ and $p$ an odd prime. 
\end{abstract}

%{\small Mathematics Subject Classification (2000), 11T22, 11B83}

\maketitle

\section{Introduction}
The problem of converting random points of a variety (e.g a curve or Jacobian of a curve) into random bits has several cryptographic applications. Such applications are key derivation functions, key exchange protocols and design of cryptographically secure pseudorandom number generators. However the binary representation of the common secret element is {\it distinguishable} from a uniformly random bit-string of the same length. Hence one has to convert this group element into a random-looking bit-string. This can be done using a deterministic extractor.

Randomness extractors are having more and more applications in computer sciences, both in theory and in applications. Randomness extractors are objects that turn {\it weak} randomness into almost {\it ideal} randomness. For example, they can be used in designing key exchange protocols and are secure pseudorandom generators in the standard model. 

Nowaday, it's a routine matter to extract randomness from a single source using arithmetic of finite fields or Elliptic curves. However, some subexponential attacks again the discret logarithm problem on some elliptic curves are known. A recommendation is to move on Jacobian varieties( i.e hyperelliptic curves) of genus less or equals to $3$. In $1989$, Koblitz N. \cite{Ko} proposed a cryptosystem based on hyperelliptic curves. Since then, hyperelliptic curves have gain lot of interest for cryptographic  applications. Furthermore, they were shown to be competitive with elliptic curves in speed and security. In \cite{GP}, the security of genus $2$ hyperelliptic curves is assumed to be similar to that of elliptic curves of the same group size. 

At this moment, several deterministic extractors for elliptic curves are known. We refer to  \cite{Ciss1},\cite{Che},\cite{F2},\cite{F3} and the references therein. In our knowledge, few work have been done on randomness extractors of Jacobian of hyperelliptic curves.  In general terms the problem can be described as follows. Given an algebraic variety $\mathcal{V}$ over $\mathbb{F}_q$ and one or several sources of random but not necessarily uniformly generated points on $\mathcal{V}$, design an algorithm to generate long strings of random bits with a distribution that
is close to uniform. In \cite{Dvir}, Dvir has considered the problem of constructing randomness extractors for algebraic varieties. His construction requires only one but rather
uniform source of points on $\mathcal{V}$. In fact, the task of extraction from an algebraic variety generalize the problem of extraction from affine sources which has drawn a considerable attention for cryptographic applications.

In this paper, we proposed an improved version of the extractors $SEJ$ and $PEJ$ from \cite{F} for $J(\mathbb{F}_q)$, where $q=p^n$ for some positive integer $n$, the Jacobian of a genus $2$ hyperelliptic curve $\mathcal{H}$ defined over $\mathbb{F}_q.$ In fact, for a given reduced divisor $D$ of $J(\mathbb{F}_q)$ we used the Mumford's representation of $D$ and extract a perfectly random bit string of the coordinate of its undeterminate corresponding to the sums of the abcissas of the rational points on $\mathcal{H}$ in the support of $D$. The element extracted from $D$ chosen randomly in $J(\mathbb{F}_q)$ is statistically close to uniform in $\mathbb{F}_q$. Instead of computing directly the statistical distance between the value of the element extracted form $D$ and a random variable in $\mathbb{F}_q$ as done by Farashahi in \cite{F}, we compute first the collision probability by summing over polynomials of degree less or equal to $2$ in $\mathbb{K}[X]$, where $\mathbb{K}$ is any subfield of $\overline{\mathbb{F}_q}.$

The remainder of this paper is organized as follows. In Section \ref{sec1}, we recall some definitions and results on the measurement parameters of randomness and bounds on character sums with polynomial arguments. In sections \ref{sec:3:1} and \ref{sec:3}, we present and analyze the security of the modified version of the randomness extractors defined over $\mathbb{F}_{p^n}$ and $\mathbb{F}_p$, respectively. We show that the outputs of these extractors, for a given uniformly random point of $\mathbb{F}_q$, are statistically close to a uniformly random variable in $\mathbb{F}_q$. For the analysis of these extractors, we need some bounds on the cardinalities of the character sums over polynomial defined over $\mathbb{K}[X].$ We give our estimates for them using Mordell's bound for polynomial of degree $\leq 2.$ 
\medskip
\medskip

\section{Preliminary results}
\label{sec1}
In this section, we recall basics definitions and notations that will be used throughout the paper.

\subsection*{Notation}: For a finite field $\mathbb{F}$, we note by $\overline{\mathbb{F}}$ the algebraic closure of the field $\mathbb{F}$. Along this paper, $\mathbb{F}_q$ is a finite field with $q$ elements where $q=p^n$, with $p$ an odd prime and $n\in\mathbb{N}^{*}.$ Let $E$ be a curve defined over $\mathbb{F}_q$ , then the set of $\mathbb{F}_q$-rational points on $E$ is denoted by $E(\mathbb{F}_q).$  Let $\mathbb{K}$ be any subfied of $\overline{\mathbb{F}_q}$ and $\mathcal{H}$ be an imaginary hyperelliptic curve, then we denote by $J_{\mathcal{H}}(\mathbb{K})$, the Jacobian of $\mathcal{H}$ over $\mathbb{K}$.  We denote by $\mathbb{K}[X]_{\leq d}$ the sets of polynomial in $\mathbb{K}[X]$ of degree less or equal  to $d.$ Further, let $lsb_k(x)$ be the $k$-least significant bits of a random element in $\mathbb{F}_q.$
\medskip

\newpage

\subsection{Hyperelliptic Curves}

\begin{definition}{Jacobian of Hyperelliptic Curves}\newline
Let $\mathcal{H}$ be an Hyperelliptic curve of $g$ in $\mathbb{F}_q$, where $q$ is odd. Here, we consider $\mathcal{H}$ to be an {\it imaginary} hyperelliptic curve.  Then $\mathcal{H}$ has a plane model of the form $y^2 = f(x)$, where $f$ is a square-free polynomial and $deg(f) = 2g+1.$ For any subfield $\mathbb{K}$ of $\overline{\mathbb{F}_q}$ containing $\mathbb{F}_q$ , the set
$$\mathcal{H}(\mathbb{K}) = \{(x, y) : x, y \in \mathbb{K}, y^2 = f(x)\}\cup \{P_\infty\},$$
is called the set of $\mathbb{K}$-rational points on $\mathcal{H}$. The point $P_\infty$ is called the point at infinity for $\mathbb{H}$. A point $P$ on H, also written $P \in \mathcal{H}$, is a point $P \in  \mathcal{H}(\mathbb{F}_q)$. The
negative of a point $P = (x, y)$ on $\mathcal{H}$ is defined as $-P = (x, -y)$ and $-P_\infty = P_\infty.$
\end{definition}
\medskip

\begin{definition}{Reduced divisors}

For each nontrivial class of divisors in $J_{\mathcal{H}}(\mathbb{K})$, there exist a unique divisor $D$ on $\mathcal{H}$ over $\mathbb{K}$ of the form
$$D=\sum_{i=1}^{g}P_i-rP_\infty$$
where $P_i = (x_i , y_i )\neq P_\infty$ , $P_i \neq -P_j$ , for $i\neq j$, and $r\leq g$. Such a divisor is called a reduced divisor on $\mathcal{H}$ over $\mathbb{K}$. By using Mumford's representation \cite{Mum}, each reduced divisor $D$ on $\mathcal{H}$ over $\mathbb{K}$ can be uniquely represented by a pair of polynomials $[u(x), v(x)], u, v \in \mathbb{K}[x],$ where $u$ is monic, $deg(v) < deg(u)\leq  g,$  and $u \mid (v^2-f)$. Precisely $u(x)= \prod_{i=1}^{r} (x -x_i )$ and $v(x_i ) = y_i$ . The neutral element of $J_{\mathcal{H}}(\mathbb{K})$, denoted by $\mathcal{O}$, is represented by $[1, 0].$ Cantor’s algorithm,\cite{Ca}, efficiently computes the sum of two reduced divisors in $J_{\mathcal{H}}(\mathbb{K})$ and expresses it in reduced form.
\end{definition}

\subsection{Measure of Randomness}

 \begin{definition}{Collision Probability}\newline
Let $\mathcal{X}$ be a finite set and $X$ an $\mathcal{X}$-valued random variable. The collision probability  of $X$  denoted by $Col(X)$, is the probability $Col(X)=Pr[X=X^{'}]=\sum_{x\in\mathcal{X}}Pr[X=x]^2.$
\end{definition}

\medskip
\begin{definition}{Statistical Distance}\newline
Let $\mathcal{X}$ be a finite set and $X$. If $X$ and $Y$ are $\mathcal{X}$-valued random variables. Then the statistical Distance between $X$ and $Y$ is define as 
$$SD(X,Y)=\frac{1}{2}\sum_{x\in\mathcal{X}}|Pr[X=x]-Pr[Y=x]|.$$

\end{definition}

Let $U_{\mathcal{X}}$ be a random variable uniformely distributed on $\mathcal{X}$  and $\delta\leq 1$ be a positive real number. Then a random variable $X$ on $\mathcal{X}$ is said to be $\delta$-uniform if $SD(X,U_{\mathcal{X}})\leq \delta.$ 

\begin{lemma}{Relation between SD and Col(X)}\newline
Let $X$ be a random variable over a finite set $\mathcal{X}$ of size $|\mathcal{X}|$ and $\Delta=SD(X,U_{\mathcal{X}})$ the statistical distance between $X$ and $U_{\mathcal{X}}$ ,  $U_{\mathcal{X}}$ is be a random variable uniformely distributed on $\mathcal{X}$. Then

\begin{equation}
\label{eq:relation1}
Col(X)\geq \frac{1+4\Delta}{|\mathcal{X}|}.
\end{equation}

\end{lemma}

\begin{definition}{Deterministic $(\mathcal{Y},\delta)$-extractor.}\newline
Let $\mathcal{X}$ and $\mathcal{Y}$ be two sets. Let $Ext$ be a function $Ext:\mathcal{X}\leftarrow \mathcal{Y}$. We say that $Ext$ is a deterministic $(\mathcal{Y},\delta)$-extractor of $\mathcal{X})$-extractor if $Ext(U_{\mathcal{Y}})$ is $\delta$-uniform on $(\mathcal{Y}$. That is, 

$$SD(Ext(U_{\mathcal{X}}),U_{\mathcal{X}})\leq \delta.$$
\end{definition}

\subsection{Character Sums with Polynomial arguments}

\begin{definition}{Character}\newline
Let $G$ be an abelian group. A character of $G$ is a homomorphisme from $G\rightarrow \mathbb{C}^{*}$. A character is trivial if it is identically $1$. We denote the trivial character by $\psi_0.$
\end{definition}

\begin{definition}
Let $\mathbb{F}_q$ be a given finite field. An additive character $\psi: \mathbb{F}_q^{+}\rightarrow \mathbb{C}$ is
a character $\psi$ with $\mathbb{F}_q$ considered as an additive group. A multiplicative character $\psi:
F^{∗}_q\rightarrow \mathbb{C}$ is a character with $\mathbb{F}_q^{∗}= \mathbb{F}_q\backslash\{0\}$ considered as a multiplicative group. We extend $\psi$ to $\mathbb{F}_q$ by defining $\psi(0) = 1$ if $\psi$ is trivial, and $\psi(0) = 0$ otherwise. Note that the extended $\psi$ still preserves multiplication.

\end{definition}

The main interests of exponential sums is that they allow to construct some charac-
teristic functions and in some cases we know good bounds for them. The use of these
characteristic functions can permit to evaluate the size of these sets. We focus on certain
character sums, those involving the character $e_p$ defined as follows.

\begin{theorem}{Multiplicative Characters of $\mathbb{F}_p$}\newline
The multiplicative characters of $\mathbb{F}_p$, where $p$ is a prime, are given by: $\forall x\in \mathbb{F}_p, e_p(x)=e^{\frac{2i\pi x}{p}}\in\mathbb{C}^{*}.$

\end{theorem}

\begin{theorem}{Additive Characters of $\mathbb{F}_q$}\newline
Suppose that $q=p^n$, where $p$ is a prime and $n\geq 1$. The additive characters of $\mathbb{F}_q$ are given by: $\psi(x)=e_p(Tr(x))$ where $Tr(x)=x+x^p+\cdots+x^{p^{n-1}}$ is the trace of $x$.
\end{theorem}

\begin{lemma}
\label{lem:size}
	Let $p$ be a prime number and $G$ a multiplicative subgroup of $\mathbb{F}_p^{*}$.
	\begin{enumerate}
	\item If $a=0$, $\sum_{x=0}^{p-1}e_p(ax)=p.$
	\item For all $a\in F_p^{*}, \sum_{x=0}^{p-1}e_p(ax)=0.$
	\end{enumerate}
\end{lemma}

\begin{proof}
See  \cite{Zim} pp $69.$
\end{proof}

\begin{theorem}{Winterhof Bound}\newline
\label{th:winterhof}
 Let $V$ be an additive subgroup of $\mathbb{F}_{p^n}$ and $\psi$ and additive character of $\mathbb{F}_{p^n}.$ Then
 
 $$\sum_{a\in\mathbb{F}_{p^n}}\Big|\sum_{x\in V}\psi(ax)\Big|\leq p^n.$$

\end{theorem}

\begin{proof}
See \cite{Win}
\end{proof}
\subsection{Elementary Bounds on character sums with polynomial arguments}.
\medskip

Here we use the same presentation as in \cite{sums}. Let $P(X)\in\mathbb{F}_q[X]$ be a polynomial of degree at most $d$. It seems reasonable to expect the distribution of  values of $P(x)$ as $x$ varies in $\mathbb{F}_q$ to be spread out of $\mathbb{F}_q$. In fact these values belong to a set $V$ with probability about $|V|/q.$
\medskip

One important way of measuring the uniformity of distribution is through the character sums:
$$\Big|\sum_{x\in\mathbb{F}_q}\psi (P(x))\Big|.$$

There are several Theorems showing that this sum is small. In our case, we will use Mordell's bound which work for abitrary polynomial with degree $\leq d.$

\begin{theorem}{(Mordell's Bound)}\newline
\label{th:mordell}
Let $\psi$ be a non trivial additive character of $\mathbb{F}_q$ and let $P(X)$ be a nonzero polynomial of degree $d<char(\mathbb{F}_q).$ Then

\begin{equation}
\label{eq:th}
\Big|\sum_{x\in\mathbb{F}_q}\psi(P(x))\Big|\leq O\Big(d\cdot q^{1-\frac{1}{2d}}\Big).
\end{equation}
\end{theorem}

%\section{Our Contribution}
\section{Extractors over Jacobian of Hyperelliptic}
\label{sec:3}

In this section, we propose an impoved version of the extractors proposed by Farashahi in \cite{F} on Jacobian of hyperelliptic curve of genus $2$ with odd characteristic. In our case, instead of working directly with points on the Jacobian $J(\mathbb{F}_q)$, we use there Mumford's representation. Therefore, our source become a subset of the polynomial ring $\mathbb{F}_q[X]$ where $q=p^n$, with $p$ an odd prime and $n\geq 1.$ Our approach uses character sums with polynomial arguments. 
\medskip

Let $J(\mathbb{F}_q)$ be the Jacobian of the hyperelliptic curve $\mathcal{H}$.We recall that each reduced divisor $D$ on $\mathcal{H}$ over $\mathbb{F}_q$ can be uniquely represented by a pair of polynomials $[u(x), v(x)], u, v \in \mathbb{F}_q[x].$ $D$ can also be uniquely represent by at most $2$ points on $\mathcal{H}$. Then, there is a map 

$$\begin{array}{l c l}
h:J(\mathbb{F}_q) & \longrightarrow & \mathbb{F}_q[x]^2\\
P+Q-2P_\infty & \longmapsto&  [x^2+u_1x+u_0,v_1x+v_0],\\
P-P_\infty & \longmapsto & [x+u_0,v_0],\\
\mathcal{O} & \longmapsto & [1,0].
\end{array}$$
Therefore, we define the $Sum$ and $Prod$ extractors as the restriction of $SEJ$ and $PEJ$ to the first component of the image of $h.$
 
\subsection{Sum  and Product Extractors for Jacobian over $\mathbb{F}_{p^n}$}.
\label{sec:3:1}
\medskip

We consider the function $f_k$ defined as follow:
$$
\begin{array}{l l c l}
f_k: & \mathbb{F}_q & \longrightarrow & \mathbb{F}_p^k \\
& x & \longmapsto & (x_1,x_2,\ldots, x_k)
\end{array}$$
where $x=(x_1,x_2,\ldots,x_n)$ with $x_i\in\mathbb{F}_p.$ 

\begin{definition}{Sum Extractor}\newline
The $Sum$ extractor for the Jacobian $J(\mathbb{F}_q)$ of $\mathcal{H}$ over $\mathbb{F}_q$ is defined as the function $Sum:\mathbb{F}_q[X]_{\leq 2}\rightarrow \mathbb{F}_p^{k}$ by

$$Sum(D)=\left\{
\begin{array}{c l l}
    f_k(-u_1) & \hbox{if} & D=[x^2+u_1x+u_0,v_1x+v_0],\\
    f_k(-u_0) & \hbox{if} & D=[x^2+u_0,v_0],\\
    0 & \hbox{if} & D=[1,0].
\end{array}
\right.$$
\end{definition}

\begin{definition}{Product Extractor}\newline
The product extractor $Prod$ for the Jacobian $J(\mathbb{F}_q)$ of $\mathcal{H}$ over $\mathbb{F}_q$ is defined as the function $Prod:\mathbb{F}_q[X]_{\leq 2}\rightarrow \mathbb{F}_p^{k}$ by

$$Prod(D)=\left\{
\begin{array}{c l l}
   f_k( u_0) & \hbox{if} & D=[x^2+u_1x+u_0,v_1x+v_0],\\
    f_k(-u_0) & \hbox{if} & D=[x^2+u_0,v_0],\\
    0 & \hbox{if} & D=[1,0].
\end{array}
\right.$$
\end{definition}

Let $A$ and $B$ be $\mathbb{F}_q$-valued random variables that are defined as $$A:=Sum(D),~~~ B:=Prod(D),$$ where $D\in J(\mathbb{F}_q).$ In the next Theorem, we show that provided the divisor $D$ is chosen uniformly in $J(\mathbb{F}_q),$ the element extracted from the divisor $D$ by $Sum$ or $Prod$ is indistinguishable from a uniformly random bit-string $\mathbb{F}_p^{k},$ with $k<n.$
\medskip

\begin{theorem}
\label{th:1}
Let $U_{\mathbb{F}_q}$ be a random variable uniformily distributed in $\mathbb{F}_q.$ Then
\begin{enumerate}
\item $\Delta(A,U_{\mathbb{F}_q})=O\Big(\frac{\sqrt{p^k}}{2\sqrt{q}(q+1)}\Big),$
\item $\Delta(B,U_{\mathbb{F}_q})=O\Big(\frac{\sqrt{p^k}}{2\sqrt{q}(q+1)}\Big).$
%\item $\not= 2$
\end{enumerate}
\end{theorem}

\begin{proof}
Let $\Psi$ be the set of all additive characters over $\mathbb{F}_q$. We put $f:=Sum$ and $G=\mathbb{F}_q[X]_{\leq 2}$. We consider the following sets.
\begin{align*}
M &= \{x_{k+1}\alpha_{k+1}+x_{k+2}\alpha_{k+2}+\cdots+x_n\alpha_n,x_i\in\mathbb{F}_p\}\subset \mathbb{F}_p^{n}.\\
\mathbb{A} &=\{u_1(x),u_2(x))\in G^2,\exists m\in M: f(u_1(x))-f(u_2(x))=m\}.
\end{align*}
$M$ is an additive subgroup of $\mathbb{F}_q$ of order $k$. Thus $|M|=p^k$ withe $k\geq 1$. Using Lemma \ref{lem:size}, we construct the following characteristic function for $\mathbb{A}$  $$\hbox{\textbf{1}}_{\mathbb{A}}=\frac{1}{p^n}\sum_{\psi\in \Psi}\psi(f(u_1(x))-f(u_2(x)-m)$$
wich is equal to $1$ if $f(u_1(x)-f(u_2(x))=m$ and $0$ otherwise. Therefore, we have that

$$|\mathbb{A}|=\frac{1}{p^n}\sum_{u_1(x)\in G}\sum_{u_2(x)\in G}\sum_{m\in M}\sum_{\psi\in \Psi}\psi(f(u_1(x))-f(u_2(x)-m).$$
Then $$Col(A)=\frac{1}{|G|^2}|\mathbb{A}|.$$ 
It's well known that the number of unitary polynomials of degree equl to $d$ in a polynomial field $\mathbb{F}_q[X]$ is $q^d$.  Thus we have that $|G|=q^2+q.$ Thus, $|G|^2= q^4+2q^3+q^2.$ Then, we have that

\begin{eqnarray}
\label{eqn:1}
Col(A) &=&\frac{1}{|G|^2 p^n}\sum_{u_1(x)\in G}\sum_{u_2(x)\in G}\sum_{m\in M}\sum_{\psi\in \Psi}\psi(f(u_1(x))-f(u_2(x)-m)\nonumber\\
&=&\frac{1}{|G|^2 p^k}p^{n-k}|G|^2 + \frac{1}{|G|^2 p^n}\sum_{u_1(x)\in G}\sum_{u_2(x)\in G}\sum_{m\in M}\sum_{\psi\neq \psi_0}\psi(f(u_1(x))-f(u_2(x)-m)\nonumber\\
&=& \frac{1}{p^k} + \frac{1}{|G|^2 p^n}\sum_{u_1(x)\in G}\sum_{u_2(x)\in G}\sum_{m\in M}\sum_{\psi\neq \psi_0}\psi(f(u_1(x))-f(u_2(x)-m)\nonumber\\
&=&\frac{1}{p^k} + \frac{1}{|G|^2 p^n}\sum_{\psi\neq \psi_0}\Big(\sum_{u_1(x)\in G}\psi(f(u_1(x))\Big) \Big(\sum_{u_2(x)\in G}\psi(-f(u_2(x)\Big) \Big(\sum_{m\in M}\psi(-m)\Big)\nonumber\\
&\leq & \frac{1}{p^k} + \frac{K^2}{|G|^2 p^n}\sum_{\psi\neq \psi_0}\Big(\sum_{m\in M}\psi(-m)\Big)
\end{eqnarray} 

where $K=\max_{\psi}\Big(\Big|\sum_{u_1(x)\in G}\psi(f(u_1(x))\Big|,\Big|\sum_{u_2(x)\in G}\psi(-f(u_2(x))\Big|\Big).$  

One note that  the sum 
$$\Big|\sum_{u_1(x)\in G}\psi(f(u_1(x))\Big|\simeq \Big|\sum_{x\in \mathbb{F}_q}\psi(P(x))\Big|$$
for a fix polynomial $P(x)\in \mathbb{F}_q[x]$. In fact, if $u_1(x)$ is of degree $2$, then $|f(u_1(x))|$ is approximatively equal to $|u^{'}_1(0)|$ and if $u_1(x)$ is of degree $1$ then $|f(u_1(x))|\simeq |u_1(0)|$. Thus, we can assume $P(x)$ to be of degree $d=1$.  Therefore, Theorem \ref{th:mordell} gives that 

\begin{equation}
\label{eq:2}
K\leq c_q\sqrt{q}
\end{equation}
where $c_q$ is the constant involved in inequality \ref{eq:th}. Therefore, combining inequality \eqref{eq:2} and Theorem \ref{th:winterhof}, inequality \ref{eqn:1} becomes

$$Col(A)\leq \frac{1}{p^k} + \frac{c_q^2q}{q^4+2q^3+q^2}=\frac{q^4+2q^3+q^2+c_q^2p^kq}{p^k(q^4+2q^3+q^2)}.$$

From Lemma \ref{lem:size}, we have that

$$\frac{1+4\Delta^2(A,U_{\mathbb{F}_q}}{p^k}\leq Col(A)\leq \frac{q^4+2q^3+q^2+c_q^2p^kq}{p^k(q^4+2q^3+q^2)}.$$
Therefore,

$\Delta(A,U_{\mathbb{F}_q})\leq \frac{c_q}{2\sqrt{p^{n-k}}(p^n+1)}=\frac{c_q\sqrt{p^k}}{2\sqrt{q}(q+1)}$, thus $\Delta(A,U_{\mathbb{F}_q})=O\Big(\frac{\sqrt{p^k}}{2\sqrt{q}(q+1)}\Big).$ This finishes the proof of $(1)$.

The proof of $(2)$ can be done in a similar way, thus we omit the details.

\end{proof}

\begin{cor}The functions $Sum$ and $Prod$ are deterministic $\Big(\mathbb{F}_p^k,O\Big(\frac{\sqrt{p^k}}{2\sqrt{q}(q+1)}\Big)\Big)-$extractor for $J(\mathbb{F}_q)$.
\end{cor}

\begin{proof}
The result of Theorem \ref{th:1} gives the proof of this corollary.
\end{proof}

\subsection{Sum  and Product Extractors for Jacobian over $\mathbb{F}_p$}
\label{sec:3}
Here we defined the $Sum$ and $Prod$ extractors as before on $\mathbb{F}_p$ where $p$ is a prime number $\geq 3$. We recall that if $I$ is an interval of integers, it's well known that 
$$\sum_{x\in \mathbb{F}_p}\Big|\sum_{\sigma \in I}e_p(x\sigma)\Big|\leq p\log_2(p).$$

\begin{definition}
We defined the extractors 
$S_k:\mathbb{F}_p[X]_{\leq 2}\rightarrow \{0,1\}^k $  by

$$S_k(D)=\left\{
\begin{array}{c l l}
    lsk_k(-u_1) & \hbox{if} & D=[x^2+u_1x+u_0,v_1x+v_0],\\
    lsb_k(-u_0) & \hbox{if} & D=[x^2+u_0,v_0],\\
    0 & \hbox{if} & D=[1,0].
\end{array}
\right. $$ and $P_k:\mathbb{F}_p[X]_{\leq 2}\rightarrow \{0,1\}^k $  by

$$P_k(D)=\left\{
\begin{array}{c l l}
    lsk_k(u_0) & \hbox{if} & D=[x^2+u_1x+u_0,v_1x+v_0],\\
    lsb_k(-u_0) & \hbox{if} & D=[x^2+u_0,v_0],\\
    0 & \hbox{if} & D=[1,0].
\end{array}
\right.$$

\end{definition}

The following Lemmas states that $S_k$ and $P_k$ are deterministic extractors for the Jacobian of the hyperelliptic curve.

\begin{lemma}
\label{lemma:1}
Let $G:=\mathbb{F}_p[X]_{\leq 2}$ and $U_G$ a random variable uniformly distributed in $G$ and $k$ a positive integer.Then 

$$\Delta(A,U_k)\ll\sqrt{\frac{2^{k}}{p}}\(1+\frac{\sqrt{\log_2(p)}}{p+1}\)$$
where $U_k$ is the uniform distribution on $\{0,1\}^k.$
\end{lemma}

\begin{proof}

Let $\delta=2^k$, $\sigma_0:=msb_{n-k}(p-1)$ and $A=S_k(D).$ We consider the set

$\mathbb{A} =\{u_1(x),u_2(x))\in G^2,\exists \sigma\leq \sigma_0, S_k(u_1(x))-S_k(u_2(x))-\delta\sigma\equiv 0\pmod{p}\}.$ Then $$Col(A)=\frac{1}{|G|^2}|\mathbb{A}|.$$ 

\begin{eqnarray}
\label{eqn:1}
Col(A) &=&\frac{1}{|G|^2 p}\sum_{u_1(x)\in G}\sum_{u_2(x)\in G}\sum_{\sigma\leq \sigma_0}\sum_{\psi\in \Psi}\psi(S_k(u_1(x))-S_k(u_2(x)-\delta\sigma)\nonumber\\
&=&\frac{\sigma_0+1}{p} + \frac{1}{|G|^2 p}\sum_{u_1(x)\in G}\sum_{u_2(x)\in G}\sum_{\sigma\leq \sigma_0}\sum_{\psi\neq \psi_0}\psi(S_k(u_1(x))-S_k(u_2(x)-\delta\sigma)\nonumber\\
&=& \frac{\sigma_0+1}{p} + \frac{1}{|G|^2 p}\sum_{u_1(x)\in G}\sum_{u_2(x)\in G}\sum_{\sigma\leq \sigma_0}\sum_{\psi\neq \psi_0}\psi(S_k(u_1(x))-S_k(u_2(x)-\delta\sigma)\nonumber\\
&=&\frac{\sigma_0+1}{p} + \frac{1}{|G|^2 p}\sum_{\psi\neq \psi_0}\Big(\sum_{u_1(x)\in G}\psi(S_k(u_1(x))\Big) \Big(\sum_{u_2(x)\in G}\psi(-S_k(u_2(x)\Big) \Big(\sum_{\sigma\leq \sigma_0}\psi(-\delta\sigma)\Big)\nonumber\\
&\leq & \frac{\sigma_0+1}{p} + \frac{K^2}{|G|^2 p}\sum_{\psi\neq \psi_0}\Big(\sum_{\sigma\leq \sigma_0}\psi(-\delta\sigma)\Big)\\
&\leq & \frac{\sigma_0+1}{p} + \frac{p\log_2(p)}{|G|^2}
\end{eqnarray}

where $K=\max_{\psi}\Big(\Big|\sum_{u_1(x)\in G}\psi(S_k(u_1(x))\Big|,\Big|\sum_{u_2(x)\in G}\psi(-S_k(u_2(x))\Big|\Big)\leq \sqrt{p}.$ 
\medskip

Therefore,
$$\Delta(A,U_k)\ll\sqrt{\frac{2^{k}}{p}}\(1+\frac{\sqrt{\log_2(p)}}{p+1}\).$$ 
\end{proof}

\begin{lemma}
\label{lemma:2}
Let $G:=\mathbb{F}_p[X]_{\leq 2}$ and $U_G$ a random variable uniformly distributed in $G$ and $k$ a positive integer.Then 

$$\Delta(A,U_k)\ll\sqrt{\frac{2^{k}}{p}}\(1+\frac{\sqrt{\log_2(p)}}{p+1}\)$$
where $U_k$ is the uniform distribution on $\{0,1\}^k.$
\end{lemma}

\begin{proof}
The proof for the extractor $P_k$ is similar to the prood of Lemma \eqref{lemma:1}.
\end{proof}

\medskip

\section{Comparison}.
\medskip

We mainly compare our result with the result of R. R. Farashahi (see \cite{F}.) In fact, Farashahi obtained a $O(\mathbb{F}_q,\frac{1}{\sqrt{q}})$-deterministic extractor by computing directly the statistical distance. His method of proof was more complicated and involved bounds of cardinalities of some curves. 

In our approach, instead of computing directly the statistical distance, we compute the collision  probability then use the inequality \eqref{eq:relation1} to obtain a sharper estimate of the statistical distance. One sees that the upper bounds obtained in Theorem \ref{th:1} are smaller than the bounds on $SEJ$ in Proposition $1$ and $PEJ$ in Corollary $2$ in \cite{F}. 

Moreover, the output of the extractor $SEJ$ in \cite{F} is a coefficient $-u_1 \in \mathbb{F}_{p^k}$ of a polynomial of degree $2$. Or, an element in $\mathbb{F}_{p^k}$ is not necessarely a uniform random-bit string. But, in our case, we extract the $k$ least significant bits of the coefficient $-u_1$ using the function $f_k$ as defined in section \ref{sec:3:1}. So, our approach gives more advantages for further applications in cryptography.

Furthermore, we have defined the extractors $Sum$ and $Prod$ on $\mathbb{F}_p$, where $p$ is an odd prime. Our results obtained in Lemma \ref{lemma:1} and Lemma \ref{lemma:2} are, in our knowledge, new results in this subject.

\medskip

\section*{Acknowledgments} This work was carried out by a financial support from the goverment of Canada's International Developpement Research Centre(IDRC) and within the framework of the AIMS research for Africa project. The author thanks Dr. Abdoul A. Ciss for useful comments and suggestions on an earlier draft of this paper.

\end{document}